%% file: Lrank10-9-13.tex
\documentclass[11pt]{amsart}
\usepackage{latexsym,amssymb,amsmath,youngtab}
\usepackage{amsmath,amsthm,amssymb,amscd,MnSymbol}
\usepackage[mathscr]{eucal}
\usepackage{verbatim}
\usepackage{color}
\input cortdefs.tex

\def\Mn{M_{\langle \nnn,\nnn,\nnn\rangle}}
\def\Mnl{M_{\langle \mmm,\nnn,\lll\rangle}}

\def\aaa{{\bold a}}\def\bbb{{\bold b}}\def\ccc{{\bold c}}

\def\mmm{\bold m}\def\nnn{\bold n}\def\lll{\bold l}

\begin{document}

\title{New lower bounds for the rank of matrix multiplication}
\author{J.M. Landsberg}
%\date{April 2010}
 \begin{abstract} The rank of the matrix multiplication operator for $\nnn\times\nnn$ matrices is one
of the most studied quantities in algebraic complexity theory. I prove   that the  rank is
at least $3\nnn^2-o(\nnn^2)$. More precisely,   for any integer $p\leq  \nnn -1$   the   rank
is at least 
$(3-\frac 1{p+1})\nnn^2-(1+2p\binom{2p}{p-1})\nnn$.   The previous lower bound, due to Bl\"aser, was $\frac 52\nnn^2-3\nnn$ (the
case $p=1$).
 The new bounds  improve  Bl\"aser's bound  for all $\nnn>84$.
I also prove lower bounds for rectangular matrices significantly better than the the previous bound.
\end{abstract}
\thanks{supported by NSF grant  DMS-1006353}
\email{jml@math.tamu.edu}
\maketitle

\section{Introduction} 
  
Let $X=(x^i_j)$, $Y=(y^i_j)$ be $\nnn\times \nnn$-matrices with indeterminant  entries. The
{\it rank}  of matrix multiplication, denoted $\bold R(\Mn)$, is the smallest number $r$  of products
$p_{\rho}=u_{\rho}(X)v_{\rho}(Y)$, $1\leq \rho\leq r$,  where $u_{\rho},v_{\rho}$ are  linear forms, such that the entries
of the matrix product $XY$ are contained in the linear span of the $p_{\rho}$. This quantity is also
called the bilinear complexity of $\nnn\times\nnn$ matrix multiplication. 
More generally, one may define the rank $\bold R(b)$  of any bilinear map $b$, 
see \S\ref{tensorsect}.

From the point of view of geometry, rank is badly behaved as it is not semi-continuous. Geometers usually
prefer to work with the {\it border rank} of matrix multiplication, which fixes the semi-continuity problem
by fiat: the border rank of  a bilinear map $b$, denoted $\ur(b)$,  is the smallest $r$ such that $b$  can be
approximated to arbitrary precision by bilinear maps of rank $r$. By definition, one has $\bold R(b)\geq \ur (b)$.
A more formal definition is given in \S\ref{tensorsect}.

Let
$M_{\langle \nnn,\mmm,\lll\rangle}$ denote the multiplication of an $\nnn\times\mmm$ matrix
by an  $\mmm\times \lll$ matrix. 
In \cite{LOsecbnd} G. Ottaviani and I gave new lower bounds for the border rank of matrix multiplication, namely,
for all $p\leq \nnn-1$, 
$\ur(M_{\langle \nnn,\nnn,\mmm\rangle })\geq \frac{2p+1}{p+1}\nnn\mmm$. Taking $p=\nnn-1$ gives the bound 
$\ur(M_{\langle \nnn,\nnn,\mmm\rangle })\geq 2\nnn\mmm-\mmm$.
In this article it will be advantageous to work with a smaller value of $p$.
The  results of \cite{LOsecbnd}  are used here to prove:

\begin{theorem}\label{mainthmrect} Let $p\leq   \nnn-1$ be a natural number. Then
$$
\bold R(M_{\langle \nnn,\nnn,\mmm\rangle })\geq  \frac {2p+1}{p+1} \nnn\mmm +\nnn^2 -(1+2p\binom{2p}{p-1})\nnn.
$$
\end{theorem}

The previous bound, due to Bl\"aser \cite{Bl2}, was  $\bold R(M_{\langle \nnn,\nnn,\mmm\rangle })\geq 2\nnn\mmm-\mmm+2\nnn-2$.  
For square matrices Theorem \ref{mainthmrect} specializes to:

\begin{theorem}\label{mainthm} Let $p\leq  \nnn-1$  be a natural number. Then
$$
\bold R(\Mn)\geq (3- \frac 1{p+1})\nnn^2-(1+2p\binom{2p }{p-1})\nnn.
$$
In particular, $\bold R(\Mn)\geq 3\nnn^2-o(\nnn^2)$.
\end{theorem}
The \lq\lq in particular\rq\rq\ follows by setting e.g.,  $p=\lfloor \sqrt{\tlog(\nnn)}\rfloor$. 

This improves Bl\"aser's bound \cite{Bl1} of  $\frac 52\nnn^2-3\nnn$ (the case $p=1$)  for all $\nnn>84$.  
Working under my direction, Alex Massarenti and  Emanuele Raviolo \cite{MR3034546,MRgoodbnd}
improved the error term in Theorem \ref{mainthm}.
In a  preprint of this article I made a mistake in computing
the error term. Unfortunately this mistake was not noticed before Massarenti and Raviolo's paper \cite{MR3034546} was published, repeating
the error, although their contribution is completely correct and their correct bound will appear  in \cite{MRgoodbnd}.

\begin{remark}If $T$ is a tensor of border rank $r$, where the approximating
curve of rank $r$ tensors limits in such a way that $q$ derivatives of the curve are used, then the
rank of $T$ is at most $(2q-1)r$, see \cite[Prop. 15.26]{BCS}.
In \cite{MR1018840} they give explicit, but very large  upper  bounds on the order of approximation $h$ needed to write a tensor of border rank
$r$ as  lying in  the $h$-jet of a curve of tensors of rank $r$. 
\end{remark} 

The language of tensors will be used throughout.  In \S\ref{tensorsect} the language of tensors is introduced and previous
work of Bl\"aser and others is rephrased in a language suitable for generalizations. In \S\ref{mmultastensor}   I describe
the equations of \cite{LOsecbnd} and give a very easy proof of a slightly weaker result than Theorem \ref{mainthmrect}.
In \S\ref{LOeqns} I express the equations in coordinates and prove Theorem \ref{mainthmrect}.
I work  over the complex numbers throughout.

\subsection*{Acknowledgement} I thank the anonymous referee for useful suggestions and C. Ikenmeyer for  
help with the exposition.

\section{Ranks and border ranks of tensors} \label{tensorsect}
 
Let $A,B,C$ be vector spaces, of dimensions $\aaa,\bbb,\ccc$ and 
with dual spaces $A^*,B^*,C^*$. That is, $A^*$ is the  space of linear maps $A\ra \BC$. Write
$A^*\ot B$ for the space of linear maps $A\ra B$ and $A^*\ot B^*\ot C$ for the  space
of bilinear maps $A\times B\ra C$. To avoid extra $*$-s, I work   with bilinear maps $A^*\times B^*\ra C$,
i.e., elements of $A\ot B\ot C$.
Let $T: A^*\times B^*\ra C$ be a bilinear map.
 One  may also consider $T$ as a linear map $T:A^*\ra B\ot C$  (and similarly with the roles of $A,B,C$ exchanged), or
as a trilinear map $A^*\times B^*\times C^*\ra \BC$.
 
  The {\it rank} of a bilinear map  $T: A^*\times B^*\ra C$, denoted $\bold R(T)$,  is the smallest $r$ such that there
exist $a_1\hd a_r\in A$, $b_1\hd b_r\in B$, $c_1\hd c_r\in C$ such that
$T(\a,\b)=\sum_{i=1}^r a_i(\a)b_i(\b)c_i$ for all $\a\in A^*$ and $\b\in B^*$. 
%Equivalently (see  e.g.,  \cite[\S 3.1]{MR2865915}), considering the map $T: A^*\ra B\ot C$,
% $\bold R(T)$ is the smallest number of rank one elements of
%$B\ot C$ needed to span a linear space containing the linear space $T(A^*)$.
The {\it border rank} of $T$, denoted $\ur(T)$, is 
the smallest $r$ such that $T$ may be written as a limit of a sequence of rank $r$ tensors.
Since the set of tensors of border rank at most $r$ is closed,
one can use polynomials to obtain lower bounds on border rank.
That is, let $P$ be a polynomial on $A\ot B\ot C$ such that $P$ vanishes on all
tensors of border rank at most $r$:   if $T\in A\ot B\ot C$ is such that  $P(T)\neq 0$,
then  $\ur(T)>r$.

The following proposition is a rephrasing of part of the proof in \cite{Bl1}:

\begin{proposition}\label{blasprop} Let $P$ be a   polynomial of degree $d$ on $A\ot B\ot C$
such that $P(T)\neq 0$ implies $\ur(T)>r$. Let $T\in A\ot B\ot C$ be a tensor such that $P(T)\neq 0$
and  $T: A^*\ra  B\ot C$ is injective.  
Then $\bold R(T)\geq r+\aaa-d$.
\end{proposition}
As stated, the proposition is useless, as the degrees of polynomials vanishing on 
  on all tensors
of border rank at most $r$ are greater than $r$. (A general tensor of border rank
$r$ also has rank $r$.) However the conclusion still holds if one can find, for a given
tensor $T$, a polynomial, or collection of polynomials on smaller spaces, such that the nonvanishing of $P$ on $T$ 
  is equivalent to the non-vanishing of the new polynomials.  Then one   substitutes the smaller degree into the statement to obtain the nontrivial
lower bound. 

In our situation, first I will show
  $P(M_{\langle \nnn,\nnn,\mmm\rangle})\neq 0$ if and only if $\tilde P (\tilde M)\neq 0$ where $\tilde M$ is
  a tensor in a smaller space of tensors and $\tilde P$ is a polynomial of lower degree than $P$, see \eqref{estar}.
More precisely, note that in the course of the proof, $B\ot C$ does not play a role, and we will see that the
relevant polynomial, when applied to matrix multiplication $M\in A\ot B\ot C=A\ot \BC^{\nnn^2}\ot \BC^{\nnn^2}$,
will not vanish if and only if a polynomial $\tilde P$ applied to $\tilde M\in A\ot \BC^{\nnn}\ot \BC^{\nnn}$
with $\tdeg(\tilde P)=\tdeg(P)/\nnn$, does not vanish, so the proof below works in this case.  
Then,  in \S\ref{LOeqns},  I  show that $\tilde P (\tilde M)\neq 0$ is implied by the non-vanishing of two polynomials of even smaller degrees.

This is why both Bl\"aser's result and the result of this paper improve
the bound of border rank by $\aaa=\nnn^2$ minus an error term, where Bl\"aser improves
Strassen's bound and I improve the bound of \cite{LOsecbnd}. (Bl\"aser shows Strassen's equations
for border rank, when applied to the 
  matrix multiplication tensor, are equivalent to the non-vanishing of three  polynomials of 
degree $\nnn$, hence the error term of $3\nnn$. See \cite[\S 11.5]{MR2865915} for an exposition.)

To prove the Proposition, we need a standard Lemma, also used in  \cite{Bl2}, which appears in
this form in \cite[Lemma 11.5.0.2]{MR2865915}:

\begin{lemma} \label{alphachoicelem} 
Let $ \BC^{\aaa}$ be given a basis.
Given  a    polynomial $P$ of degree $d$ on $\BC^{\aaa}$, there exists 
a set of at most  $d$ basis vectors such that 
$P$ restricted to their span is not identically zero.
\end{lemma}

The lemma follows by simply choosing a monomial that appears in $P$, as it can involve at most $d$ basis vectors.

\begin{proof}[Proof of Proposition  \ref{blasprop}]
Let $\bold R(T)=r$ and assume we have written $T$ as a sum of $r$ rank one tensors.
Since $T: A^*\ra  B\ot C$ is injective    we may  write $T=T'+T''$
with $\bold R(T')=\aaa$, $\bold R(T'')= r-\aaa$
and $ T': A^*\ra  B\ot C$    injective.
Now consider the $\aaa$ elements of $A\ot B\ot C$ appearing in
$T'$. Since they are linearly independent,
by Lemma \ref{alphachoicelem}
we may choose a subset of $d$ of them  such that $P$,  evaluated on 
the sum of terms in $T$ whose $A$ terms are in the span of these $d$ elements,  is not identically zero.
  Let $T_1$ denote the sum of the terms in $T'$
not involving the (at most) $d$ basis vectors needed for nonvanishing, so $\bold R (T_1)\geq \aaa- d$.
Let $T_2=T-T_1+T''$.
 Now $\ur(T_2)\geq r$ because  $P(T_2)\neq 0$.   Finally $\bold R(T)=\bold R(T_1)+\bold R(T_2)$.
\end{proof}

Let $G(k,V)\subset \BP \La k V$ denote the Grassmannian of $k$-planes through the origin in $V$ in its Pl\"ucker embedding.
That is, if a $k$ plane is spanned by $v_1\hd v_k$, we write it as $[v_1\ww\cdots \ww v_k]$. 
One says a function on $G(k,V)$ is a   polynomial of degree $d$ if,
as a function in the Pl\"ucker coordinates, it is a degree $d$   polynomial.
The Pl\"ucker coordinates $(x^{\mu}_{\a})$, $k+1\leq \mu\leq \tdim V=\bv$, $1\leq \a\leq k$ are obtained by choosing
a basis $e_1\hd e_{\bv}$ of $V$, centering the coordinates at $[e_1\ww\cdots \ww e_k]$,
and writing a nearby $k$-plane as $[(e_1+\sum x^{\mu}_1e_{\mu})\ww \cdots \ww (e_k+\sum x^{\mu}_ke_{\mu})]$. If the polynomial
is also homogeneous in the $x^{\mu}_{\a}$, this is equivalent to   it being  the restriction of
a homogeneous degree $d$ polynomial on $\La kV$. (The ambiguity of the scale does not matter as we are only concerned with
its vanishing.) 

\begin{lemma} \label{alphachoicelemg} 
Let $A$ be given a basis.
Given  a   homogeneous polynomial of degree $d$ on the Grassmannian $G(k,A)$, there exists 
at least $dk$ basis vectors    such that, denoting their (at most) $dk$-dimensional span by $A'$, 
$P$ restricted to $G(k,A')$ is not identically zero.
\end{lemma}
\begin{proof}
Consider the map $f:A^{\times k}\ra G(k,A)$ given by $(a_1\hd a_k)\mapsto [a_1\ww\cdots \ww a_k]$. Then $f$ is surjective.
Take the polynomial $P$ and pull it back by $f$. (The pullback $f^*(P)$ is defined by
$f^*(P)(a_1\hd a_k):=P(f(a_1\hd a_k))$.)  The pullback is   of degree $d$ in each copy of $A$. (I.e., fixing $k-1$ parameters,
it becomes a degree $d$ polynomial in the $k$-th.) Now simply apply
Lemma \ref{alphachoicelem} $k$ times to see that the pulled back polynomial is not identically zero restricted to $A'$, and thus
$P$ restricted to $G(k,A')$ is not identically zero.
\end{proof} 

\begin{remark} The bound in Lemma \ref{alphachoicelemg} is  sharp, as give $A$ a basis $a_1\hd a_{\aaa}$ and  consider the polynomial on $\La k A$ with
coordinates $x^I=x^{i_1}\hd x^{i_k}$ corresponding to the vector $\sum_I x^Ia_{i_1}\ww \cdots \ww a_{i_k}$:
$P=x^{1\hd k}x^{k+1\hd 2k}\cdots x^{(d-1)k+1\hd dk}$.
Then $P$ restricted to $G(k, \langle a_1\hd a_{dk}\rangle )$ is non-vanishing but there is no smaller subspace spanned
by basis vectors on which it is non-vanishing.
\end{remark}

\section{Matrix multiplication and its rank}\label{mmultastensor}

 Let  $\Mnl:  Mat_{\mmm\times \nnn}
 \times Mat_{\nnn\times \lll}\ra Mat_{\mmm\times\lll}$ 
 denote the matrix multiplication operator. Write $M=\BC^{\mmm}$, $N=\BC^{\nnn}$ and $L=\BC^{\lll}$.
Then 
$$
\Mnl :(N\ot L^*) \times (L\ot M^*)\ra N\ot M^*
$$
has the   interpretation as $\Mnl =Id_N\ot Id_M\ot Id_L\in (N^*\ot L ) \ot (L^*\ot M )\ot( N\ot M^*)$,
where $Id_N\in  N^*\ot N$ is the identity map. 
If one thinks of $\Mnl$ as a trilinear map $(N\ot L^*) \times (L\ot M^*) \times ( N\ot M^*)\ra \BC$,  in bases it is
$(X,Y,Z)\mapsto \ttrace(XYZ)$.
If one thinks of $\Mnl$ as a linear map $N\ot L^*\ra (L^*\ot M) \ot ( N\ot M^*)$ it is just the identity map tensored with $Id_M$.
  In particular,
if   $\a\in N \ot L^* $ is of   rank $q$, its image, considered as a linear map $L\ot M^*\ra N\ot M^*$,  is of  rank $q\mmm$.

\medskip

Returning to general tensors $T\in A\ot B\ot C$, from now on assume $\bbb=\ccc$. 
When $T=\Mnl$, one has    $A= N^*\ot L $, $B= L^*\ot M $, $C=  N\ot M^*$, so $\bbb=\ccc$ is equivalent to $\lll=\nnn$.
%For any tensor $T\in A\ot B\ot C$, if there exists $\a\in A^*$ such that $T(\a)$ is of maximal rank $\bbb$, then
%one may use the linear map $T(\a): B^*\ra C$ to identify $C\simeq B^*$, and consider $T(A^*)\subset B\ot B^*$ as a subspace
%of the space of  linear maps $B\ra B$. 

\medskip

The equations of \cite{LOsecbnd} are as follows: given $T\in A\ot B\ot C$, with $\bbb=\ccc$, take $A'\subset A$ of dimension $2p+1\leq \aaa$.
Define a  linear map
\be\label{LOmap}
T^{\ww p}_{A'}: \La p A'\ot B^*\ra \La{p+1}A'\ot C
\ene 
by first considering
$T|_{A'\ot B\ot C} : B^*\ra A'\ot C$ tensored with the identity map on $\La p A'$,
which is a map $  \La p A\ot B^* \ra \La p A \ot A\ot C$,
and then projecting the image to $\La{p+1}A'\ot C$.
Then if the determinant of this linear map is nonzero, the border rank of $T$ is at least $\frac{2p+1}{p+1}\bbb$.
If there exists an $A'$ such that the determinant is nonzero, we may think of the determinant
as a nontrivial homogeneous polynomial of degree $\binom{2p+1}p\bbb$ on $G(2p+1,A)$.

\medskip

Now consider the case $T=M_{\langle \nnn,\nnn,\mmm\rangle}$, and recall that $B=L^*\ot M$, $C=N\ot M^*$. The map 
$(M_{\langle \nnn,\nnn,\mmm\rangle})^{\ww p}_{A'}: \La p A'\ot L \ot M^*\ra \La{p+1}A'\ot N\ot M^*$ 
is actually a reduced map 
\be\label{estar}
\tilde M^{\ww p}_{A'}: \La p A'\ot L \ra \La{p+1}A'\ot N 
\ene 
tensored with
the identity map $M^*\ra M^*$, and thus its determinant is non-vanishing if and only if the determinant
of $\tilde M^{\ww p}_{A'}$ is nonvanishing. {\it But this is a polynomial of degree $\binom {2p+1}p\nnn<<\binom{2p+1}p \nnn^2$
on $G(2p+1,\nnn^2)$.}
Proposition \ref{blasprop} with $d=\binom{2p+1}p\nnn$,  $\aaa=\nnn^2$ and $r=\frac{2p+1}{p+1}\mmm\nnn$,  combined with Lemma \ref{alphachoicelemg} gives the bound
$$
\bold R(M_{\langle \nnn,\nnn,\mmm\rangle})\geq \frac{2p+1}{p+1}\nnn\mmm +\nnn^2   -(2p+1) \binom {2p+1}p\nnn.
$$
Note that this already gives the $3\nnn^2-o(\nnn^2)$ asymptotic lower bound. 
The remainder of the paper is dedicated to improving the error term. 
   
\section{The equations of \cite{LOsecbnd} in coordinates}\label{LOeqns}

Let  $\aaa=3$ (so $p=1$) and $\bbb=\ccc$, the   map \eqref{LOmap}  expressed in bases is a $3\bbb\times 3\bbb$ matrix.
If $a_0,a_1,a_2$ is a basis of $A$ and one chooses bases of $B,C$, then   elements of $B\ot C$ may be written as matrices, 
and  $T=a_0\ot X_0 + a_1\ot X_1 + a_2\ot X_2$,  where the $X_j$ are size $\bbb$ square matrices. Order the basis of
$A$ by $a_0,a_1,a_2$ and of $\La 2 A$ by $a_1\ww a_2,a_0\ww a_1,a_0\ww a_2$.
We compute
\begin{align*}
T_A^{\ww 1}(a_0\ot \b)&= \b(X_0)\ot a_0\ww a_0+\b(X_1)\ot a_1\ww a_0+\b(X_2)\ot a_2\ww a_0=-\b(X_1)\ot a_0\ww a_1-\b(X_2)\ot a_0\ww a_2,\\
 T_A^{\ww 1}(a_1\ot \b)&= \b(X_0)\ot a_0\ww a_1-\b(X_2)\ot a_1\ww a_2,\\
 T_A^{\ww 1}(a_2\ot \b)&= \b(X_0)\ot a_0\ww a_2+\b(X_1)\ot a_1\ww a_2,
\end{align*}
so  the corresponding matrix for $T_A^{\ww 1}$ is the block matrix
$$
Mat(T_A^{\ww 1})=\begin{pmatrix}
0 & -X_2 &  X_1\\ -X_1 & X_0 & 0\\ -X_2&0&X_0
\end{pmatrix}.
$$

 Now assume $X_0$ is invertible and change bases such that it is
the identity matrix. Recall the formula for block matrices
\be\label{detiden}
\tdet \begin{pmatrix} X&Y\\ Z&W\end{pmatrix}=
\tdet(W)\tdet(X-YW\inv Z),
\ene
assuming $W$ is invertible.
  Then, using the $(\bbb,2\bbb)\times (\bbb,2\bbb)$ blocking (so $X=0$ in \eqref{detiden}) 
$$
\tdet Mat(T_A^{\ww 1})= \tdet(X_1X_2-X_2X_1)=\tdet([X_1,X_2]).
$$
When  $\tdim A>3$,  if    there exists a three dimensional subspace $A'$  of $A$, such that 
  $\tdet Mat(T_{A'}^{\ww 1})\neq 0$, then   $\ur(T)\geq \frac 32\bbb$ as this is \eqref{LOmap}
  in the case $p=1$. These are Strassen's equations \cite{Strassen505}.

\medskip
I now phrase the equations of \cite{LOsecbnd} in coordinates. Let $\tdim A=2p+1$. 
Write $T=a_0\ot X_0+\cdots + a_{2p}\ot X_{2p}$. The expression of \eqref{LOmap}  in  bases 
is as follows: write
$a_I:=a_{i_1}\ww \cdots \ww a_{i_p}$ for $\La pA $,   require that    the first
$\binom{2p}{p-1}$ basis vectors have $i_1=0$, that  the second $\binom{2p}{p}$ do not, and  call these multi-indices
$0J$ and $K$. Order the bases of $\La{p+1}A$ such that the first $\binom{2p}{p+1}$ multi-indices do
not have $0$, and the second $\binom{2p}{p}$ do, and furthermore that the second set of indices is ordered
the same way as $K$, only we write $0K$ since a zero index is included.
Then the resulting matrix
is of the form
\be\label{taeqn}
\begin{pmatrix}
0 & Q\\ \tilde Q & R 
\end{pmatrix}
\ene
where this matrix is blocked $(\binom{2p}{p+1} \bbb, \binom{2p}p\bbb)\times (\binom{2p}{p+1}\bbb,\binom{2p}p\bbb)$, 
$$
R=
\begin{pmatrix}
X_0 & &  \\
  &\ddots   & \\
& &  X_0
\end{pmatrix}, 
$$
and   $Q,\tilde Q$ have entries  in blocks consisting of $X_1\hd X_{2p}$ and zero. 
  Thus if $X_0$ is the identity matrix, so is $R$ and the determinant equals the determinant of $Q\tilde Q$.
If $X_0$ is the identity matrix,  when $p=1$ we have $Q\tilde Q=[X_1,X_2]$ and when $p=2$
\be\label{5matrix}
Q\tilde Q= \begin{pmatrix}
0&  [X_1,X_2] & [X_1,X_3]& [X_1,X_4]\\
 [X_2,X_1]&0&[X_2,X_3]& [X_2,X_4]\\
 [X_3,X_1]&[X_3,X_2]&0& [X_3,X_4]\\
 [X_4,X_1]&[X_4,X_2]& [X_4,X_3]&0
\end{pmatrix}.
\ene

In general, when $X_0$ is the identity matrix,  $Q\tilde Q$ is a 
    block
 $\binom {2p}{p-1}\bbb\times \binom{2p}{p-1} \bbb$ matrix
whose   block entries are either zero or commutators $[X_i,X_j]$.

To prove Theorem \ref{mainthmrect} we work with $\tilde M_{A'}^{\ww p}$ of \eqref{estar}, so $\bbb=\nnn$.  
First apply Lemma \ref{alphachoicelem} to  choose $\nnn$ basis vectors such that restricted to them
$\tdet(X_0)$ is non-vanishing, and then we consider our polynomial $\tdet(Q\tilde Q)$  as defined on $G(2p,(2p+1)\nnn^2-1)$, 
and apply Lemma \ref{alphachoicelemg}, using 
$2p\binom{2p}{p-1}\nnn$ basis vectors to insure it is non-vanishing. Our error term is thus $\nnn+2p\binom{2p}{p-1}\nnn$, and the theorem follows. 

%\begin{remark} It would be interesting to have direct  geometric interpretations of the hypersurface 
%in the Grassmannian $G(2p,\fsl(B))$ given by $\tdet(Q\tilde Q)$. The derivation shows that  
% the vanishing of the  determinant is independent of our choice of basis.
%\end{remark}

\begin{remark}\label{qextra}
 In \cite{MR3034546,MRgoodbnd}, they show the matrix $Q\tilde Q$ can be made to have a nonzero determinant by a   subtle combination
of factoring and splitting it into a sum of two matrices that carries a lower cost than just taking its determinant.
\end{remark}

 \bibliographystyle{amsplain}
 
\bibliography{Lmatrix}

\end{document}

%% file: cortdefs.tex
\newtheoremstyle{custom}% name
  {3pt}%      Space above
  {3pt}%      Space below
  {\slshape}%         Body font
  {}%         Indent amount (empty = no indent, \parindent = para indent)
  {\bfseries}% Thm head font
  {.}%        Punctuation after thm head
  { }%     Space after thm head: " " = normal interword space;
   {}%         Thm head spec (can be left empty, meaning `normal')
\theoremstyle{custom}
\newtheorem{theorem}{Theorem}[section]
\newtheorem{proposition}[theorem]{Proposition}

\newtheorem{proposition/definition}[theorem]{Proposition/Definition}
\newtheorem{lemma}[theorem]{Lemma}

\theoremstyle{definition}

\theoremstyle{remark}
\newtheorem{remark}[theorem]{Remark}

\newtheoremstyle{exercise}% name
  {3pt}%      Space above
  {6pt}%      Space below
  {}%         Body font
  {}%         Indent amount (empty = no indent, \parindent = para indent)
  {\bfseries}% Thm head font
  {:}%        Punctuation after thm head
  { }%     Space after thm head: " " = normal interword space;
   {}%         Thm head spec (can be left empty, meaning `normal')
\theoremstyle{exercise}
\newtheorem{exercise}[theorem]{Exercise}
\newtheoremstyle{exercises}% name
  {3pt}%      Space above
  {6pt}%      Space below
  {}%         Body font
  {}%         Indent amount (empty = no indent, \parindent = para indent)
  {\bfseries}% Thm head font
  {:}%        Punctuation after thm head
  {\newline}%     Space after thm head: " " = normal interword space;
   {}%         Thm head spec (can be left empty, meaning `normal')

\theoremstyle{exercise}
\newtheorem{exercises}[theorem]{Exercises}

\input epsf
\def\boxit#1{\vbox{\hrule height1pt\hbox{\vrule width1pt\kern3pt
  \vbox{\kern3pt#1\kern3pt}\kern3pt\vrule width1pt}\hrule height1pt}}

\def\bv{\bold v}

\def\BC{\mathbb C}

\def\BP{\mathbb P}

\def\tdim{{\rm dim}}

\def\hd{,...,}
\def\ww{\wedge}

\def\inv{{}^{-1}}

\def\11{\mathbf 1}

\def\a{\alpha}

\def\b{\beta}

\def\ot{{\mathord{ \otimes } }}

\def\ra{{\mathord{\;\rightarrow\;}}}

\def\La#1{\Lambda^{#1}}

\def\a{\alpha}
\def\b{\beta}

\def\BP{\mathbb  P}
\def\BC{\mathbb  C}

\def\hd{, \hdots ,}

\def\inv{{}^{-1}}

\def\La#1{\Lambda^{#1}}

\def\ur{\underline {\bold R}}

\def\ra{\rightarrow}

\def\tdeg{\operatorname{deg}}
\def\tdet{\operatorname{det}}

\def\ttrace{\operatorname{trace}}

\def\tdim{\operatorname{dim}}

\def\ww{\wedge}

\def\bbb{{\bold{b}}}

\def\be{\begin{equation}}
\def\ene{\end{equation}}
\def\aaa{{\bold {a}}}
\def\bbb{{\bold {b}}}
\def\ccc{{\bold {c}}}

\DeclareMathOperator{\tlog}{log}